\documentclass[runningheads]{llncs}
\usepackage{amssymb}
\usepackage{epsfig}
\usepackage{amsmath}
\usepackage{amssymb}
\usepackage{color}
\newcommand{\au}{{\bf a}}
\newcommand{\FF}{{\mathbb F}}

\newcommand{\C}{{\cal C}}

\newcommand{\lo}{\overline \lambda}
\newcommand{\lu}{{\underline \lambda}}
\newcommand{\cc}{{\bf c}}
\newcommand{\remove}[1]{}

\newcommand{\x}{{\bf x}}
\newcommand{\HC}{{\cal H}}

\newcommand{\cF}{{\cal F}}

\newcommand{\cB}{{\cal B}}

\begin{document}
\titlerunning{Group Testing  in Arbitrary Hypergraphs}
\authorrunning{A. De Bonis}
\title{Group Testing  in Arbitrary Hypergraphs and Related Combinatorial Structures}
\author{ Annalisa De Bonis}
\institute{
Dipartimento di Informatica, Universit\`a di Salerno,
 Fisciano (SA), Italy 
\email{adebonis@unisa.it}}

\maketitle
\begin{abstract} 
We consider  a generalization of {\em group testing} where the potentially contaminated sets are 
the members of a given hypergraph $\cF=(V,E)$. This generalization finds  application in contexts where contaminations can be conditioned by some kinds of social and geographical clusterings.
We study non-adaptive algorithms, two-stage algorithms, and three-stage algorithms. Non-adaptive group testing algorithms are algorithms in which all tests are decided beforehand and therefore can be performed in parallel, whereas two-stage and three-stage group testing algorithms consist of 
two stages and three stages, respectively, with each stage being a  completely non-adaptive algorithm.
In {\em classical} group testing, the potentially infected sets are all subsets of up to a certain number of elements of the given input set. 
For classical group testing, it is known that there exists a correspondence between   non-adaptive algorithms and  superimposed codes, and between two-stage group testing and disjunctive list-decoding codes and selectors. 
Bounds on the number of tests for those algorithms are derived from the bounds on the dimensions of the corresponding combinatorial structures. Obviously, the upper bounds for the classical case apply also to our group testing model. In the present paper, we aim at improving on those upper bounds by leveraging on the characteristics of the particular hypergraph at hand. 
In order to cope with our version of the problem, we introduce new combinatorial structures that generalize the notions of {\em classical} selectors and superimposed codes.
\end{abstract}
\section{Introduction}
Group testing consists in detecting the defective members of a set of objects $O$ by performing tests on properly chosen subsets ({\em pools}) of the given set $O$.  A test yields a ``yes" response if the tested pool contains one or more defective elements, and a ``no" response otherwise. The goal is to find all defectives by using as few tests as possible. Group testing origins date back to World War II when it was introduced as a possible technique for mass blood testing  \cite{dor}.  Since then, group testing has become a  fundamental problem in computer science and has been widely investigated in the literature both for its practical implications and for the many theoretical challenges it poses. Applications of group testing span a wide variety of situations ranging  from conflict resolution algorithms for multiple-access systems \cite{wolf},  fault diagnosis in optical networks \cite{optical}, quality control in product testing \cite{sobel}, failure  detection in wireless sensor networks \cite{faultSensor}, data compression \cite{hl}, and many others. 
Among the modern applications of group testing,  some of the most important ones are related to the field of  molecular biology, where group testing is especially employed in the design of  screening experiments. Du and Hwang  \cite{dh1}  provide an extensive coverage of the most relevant applications of group testing  in this area. 
\remove{
The different contexts to which group testing applies often call for variations of the classical model that best adapt to the characteristics of the problems. These variants concern the test model   \cite{jcb}, \cite{tcs3},  \cite{TCS},  \cite{it}, the number of pursued defective elements  \cite{majority}, \cite{disjointGT}, as well as the structure of the test groups  \cite{disjointGT},   \cite{tcs1} and the number of contaminated test groups \cite{jco}.
}
In classical  group testing, the set of defectives is any of the possible subsets of size less than or equal to a certain parameter $d$.
In the present paper, we consider a more general version of group testing parameterized by a hypergraph $\cF=(V,E)$,
 with the contaminated set being one of the hyperedges of $E$. 

\smallskip
\par\noindent{\bf Related work  and our contribution.}
Classical group testing has been studied very thoroughly both with respect to {\em adaptive strategies}, i.e. algorithms that at each step decide which group to test by looking at the responses of previous tests, and with respect to {\em non-adaptive} strategies, i.e., strategies in which all tests are decided beforehand.
It is well-know that the best adaptive strategies achieve the information theoretic lower bound $\Omega(d\log(n/d))$, 
where $n$ is the total number of elements and $d$ is the upper bound on the number of defectives.
Further, it is known that
non-adaptive strategies for classical group testing are much more costly than their adaptive counterparts. 
The  minimum number of tests used by the non-adaptive procedures is estimated by the minimum length of certain combinatorial structures known under the name of
 $d$-{\em cover free families}, or equivalently, $d$-{\em superimposed codes} and {\em strongly selective families} \cite{cms},  \cite {dyry}, \cite{erff}, \cite{kts}.
The known bounds for these combinatorial structures \cite{alon}, \cite{dyry}, \cite{rus} imply that the number of tests of any non-adaptive group testing algorithm is lower bounded by
$\Omega((d^2/ \log d) \log n)$ and that there exist non-adaptive group testing algorithms that use $O(d^2\log n)$ tests. Of particular practical interest in the field
of biological screening are two-stage group testing procedures, i.e., procedures consisting of two stages each of which is a completely non-adaptive algorithm.
As far as it concerns classical group testing, the best two-stage algorithms  \cite{siam} are asymptotically as efficient as the best completely  adaptive strategy in that they achieve the information theoretic lower bound $\Omega(d\log(n/d))$. 

In this paper, we focus on non-adaptive group testing and two-stage group testing. 
 In particular, we are interested in {\em trivial} two-stage group testing algorithms, i.e., two-stage algorithms that in the second stage are allowed only to perform tests on individual elements. These algorithms are very useful in practice since in many applications confirmatory tests on individual elements are needed anyway to ensure that they are really defective.
 We also give existential results for three-stage group testing algorithms, i.e., algorithms that consist of three stages, each of which is a non-adaptive algorithm. 
 
 The version of group testing  considered in the present paper  has been initiated only recently in \cite{edge1} and continued in \cite{edge2}.
Similar search models were previously considered by the authors of  \cite{niko} who
assumed a known community structure in virtue of which the population is partitioned into separate families and the defective hyperedges are those that contain elements from a certain number of families.
 While in that paper the information on the structure of potentially infected groups is used to improve on the efficiency of the group testing algorithms, other papers, 
  \cite{goenka}, \cite{zrb},
 exploit this knowledge to improve on the efficiency of decoding the tests' responses. 
 A formal definition of group testing in general hypergraphs  has been given in \cite{edge1}.
 The authors of \cite{edge1} consider both adaptive and non-adaptive group testing. For the adaptive setting, when hyperedges in $E$ are  
of size exactly $d$, they give an $O(\log{|E|}+d\log^2{d})$ algorithm that is close to the  $\Omega(\log{|E|}+d)$  lower bound they also prove in the paper.
In the non-adaptive setting, they exploit a random coding technique to prove an $O(\frac dp\log{|E|})$ bound on the number of tests, 
where $d$ is the maximum size of a set $e \in E$ and $p$ is a lower bound on the size of the difference $e'\setminus e$ between any two hyperedges  $e,e' \in E$. 
In \cite{edge2} the author presents a new adaptive algorithm for generalized group testing, which is
asymptotically optimal if $d = o(\log_2 |E|)$  and, for $d = 2$,  gives an asymptotically
optimal  algorithm that works in three stages.

In this paper, we  formally define  combinatorial structures that are substantially equivalent to non-adaptive algorithms  for group testing in general hypergraphs. 
The combinatorial structures introduced in this paper
 extend the above-mentioned classical superimposed codes \cite{kts} in a way such that  the desired property is not enforced for all subsets of up to a certain number
 of codewords 
  but  only for subsets of codewords associated to the hyperedges of the given  hypergraph. 
Constructions for these generalized superimposed codes allow to achieve, in the non-adaptive setting, the same $O(\frac dp\log{|E|})$ upper bound obtained  in \cite{edge1}. 
In order to design our algorithms, we introduce a notion of selectors, also parameterized by the  set of hyperedges $E$, that generalizes the notion of $(k,m,n)$-selector introduced in \cite{siam}. For particular values of the involved parameters, our selectors correspond to non-adaptive algorithms for our group testing problem, These selectors are also at the basis of our two-stage and three-stage algorithms. In particular, we give a trivial two-stage algorithm that uses $O(
\frac {qd}\chi \log |E|+dq)$ tests, where $\chi$ is a lower bound on  $|\bigcup_{i=1}^q e'_i\setminus e|$, for any $q+1$  distinct hyperedges $e,e'_1\ldots,e'_q$. If for 
a given constant  $q\geq 1$, it holds that  $|\bigcup_{i=1}^q e'_i\setminus e|=\Omega(d)$, then this algorithm asymptotically achieves the above-mentioned $\Omega(\log |E|+d)$ lower bound.
Further, we give an $O(\sqrt d\log |E|+d)$ three-stage group testing algorithm.
\remove{
On the other hand, we remark that the notion of $E$-separable code we introduce here, is implicitly employed
in the design of the non-adaptive algorithms
of \cite{edge_1} and \cite{edge_2}. Both papers provide an $O(d\log E)$ upper bound on the number of tests used by their non-adaptive algorithm. Here,we improve on that upper bound and provide an almost matching lower bound by providing a lower bound on the length of $E$-separable codes.}

\section{Notations and terminology}\label{sec:notation}
For any positive integer $m$, we denote by $[m]$ the set of integers $\{1,\ldots,m\}$. 
A hypergraph is
a pair $\cF=(V,E)$, where 
$V$ is a finite set and $E$ is a family of subsets of $V$. The elements of $E$ will be called
hyperedges. If all hyperedges of $E$ have the same size $d$ then the hypergraph is said to be $d$-uniform.

In the present paper,  the hypergraph specifying the set of potentially conta\-minated sets is assumed to have $V=[n]$.
Let $\cF=([n],E)$ be a given hypergraph, and 
 let $\chi =\min\{ |\bigcup_{i=1}^q e'_i\setminus e|,  \mbox{ for any $q+1$ distinct $e,e'_1,\ldots,e'_q\in E$}\}$.
For any  hyperedge $e\in E$ and any $q$ distinct hyperedges $e'_1,\ldots,e'_q\in E\setminus\{e\}$, we denote
by 
 $S_{e,(e'_1,\ldots,e'_q)}$  the set of $d+\chi$ integers in $[n]$ such that $d$ of these integers are the elements of $e$
 whereas the remaining are the $\chi$ smallest elements in $\bigcup_{i=1}^q e'_i\setminus e$.
  In order to properly define the combinatorial structures related to our group testing problem, we extend the definition of   $S_{e,(e'_1,\ldots,e'_q)}$ to  arbitrary hypergraphs without assuming any lower bound on the 
 size of the unions $\bigcup_{i=1}^q e'_i\setminus e$. To this aim,  we augment the set of vertices $[n]$ with $\chi$ dummy vertices $i_{n+1},\ldots i_{n+\chi}$, with $\chi$ this time being any positive integer smaller than or equal to $n-d$, and define
 $S_{e,(e'_1,\ldots,e'_q),\chi}$ as follows:  $S_{e,(e'_1,\ldots,e'_q),\chi}$ is defined exactly as  
$S_{e,(e'_1,\ldots,e'_q)}$ if $|\bigcup_{i=1}^q e'_i\setminus e|\geq \chi$; otherwise it is defined as the subset of $d+\chi$ integers in $[n]$ such that $d$ of these integers are the  vertices of $e$ and the remaining are the integers in 
$\left(\bigcup_{i=1}^q e'_i\setminus e\right) \cup\{ i_{n+1},\ldots,i_{n+\chi-|\bigcup_{i=1}^q e'_i\setminus e|}\}$. 
 
 For $q=1$, we will denote  $S_{e,(e'_1,\ldots,e'_q)}$ and $S_{e,(e'_1,\ldots,e'_q),\chi}$ with $S_{e,e'}$ and $S_{e,e',\chi}$, respectively. Notice that given two hyperedges $e,e'\in E$, the set $S_{e,e',\chi}$ contains  one or more dummy vertices if and only if $|e'\setminus e|<\chi$.
 
Notice that there are at most $|E|{|E|-1\choose q}$  distinct sets $S_{e,(e'_1,\ldots,e'_q)}$  since for a given $(q+1)$-tuple of distinct hyperedges $(e,e'_1,\ldots,e'_{q+1})$ 
 any permutation of $e'_1,\ldots,e'_{q+1}$ does not affect $S_{e,(e'_1,\ldots,e'_q)}$. Moreover, for any positive $\chi$, this estimate holds also for  the number of distinct sets $S_{e,(e'_1,\ldots,e'_q),\chi}$ since the indices of the dummy vertices that are eventually added to $S_{e,(e'_1,\ldots,e'_q)}$, in order to obtain $S_{e,(e'_1,\ldots,e'_q,\chi)}$, are fixed once the hyperedges $e,e'_1,\ldots,e'_q$ have been chosen. More precisely, they are the smallest $\chi -|\bigcup_{i=1}^q e'_i\setminus e|$ integers in   $\{i_{n+1},\ldots i_{n+\chi}\}$.

We remark that in  our group testing problem, given an input hypergraph $\cF=([n],E)$, every vertex of $[n]$ is contained in at least one
 hyperedge of $E$. If otherwise, one could remove the vertex from the hypergraph without changing the collection of potentially defective hyperedges.
 As a consequence, for a given hypergraph  $\cF=([n],E)$  we need only to specify its  set of hyperedges $E$ to characterize  both the input of the problems and the related combinatorial tools. 
 
 \section{Non-adaptive  group testing for general hypergraphs}\label{sec:ndgt}
In this section, we illustrate the correspondence between non-adaptive group testing algorithms for  an input set of size $n$ and  families of $n$ subsets. Indeed, given a family  $\FF=\{F_1,\ldots,  F_n\}$ with $F_i\subseteq [t]$, we design a non-adaptive group testing strategy as follows.  We denote the elements in the input set by the integers in $[n]=\{1,\ldots, n\}$, and for $i=1,\ldots, t$, we define the group $T_i=\{j\,:\, i\in F_j\}$. Obviously,     $T_1,\ldots,T_t$ can be tested in parallel and therefore  the resulting algorithm is non-adaptive. Conversely, given a non-adaptive group testing strategy for an input set of size $n$ that tests $T_1,\ldots,T_t$, we define a family $\FF=\{F_1,\ldots,F_n\}$  by setting $F_j=\{i\in [t]\,:\, j\in T_i\}$, for $j=1,\ldots, n$. 
Alternatively, any non-adaptive group testing algorithm for an input set of size $n$ that performs $t$ tests can be represented by a binary code of size $n$ with each codeword being a binary vector of length $t$. This is due to the fact that any family of size $n$ on the ground set $[t]$ is associated with the binary code of length $t$ whose codewords are the characteristic vectors of the members of the family. Given such a binary code  $\C=\{\cc_1,\ldots,\cc_n\}$, one has that  $j$ belongs to  pool $T_i$ if and only if the $i$-th entry $\cc_j(i)$ of $\cc_j$ is equal to 1. Such a code can be represented by a $t\times n$ binary matrix $M$ such that 
$M[i,j]=1$ if and only if element $j$ belongs to $T_i$.
 We represent the  responses to tests on $T_1,\ldots,T_t$ by a binary vector whose $i$-th entry is equal to 1 if and only if $T_i$ tests  positive. We call this vector the {\em response vector}. For any input set of hyperedges  $E$ on $[n]$, the response vector is the bitwise $OR$ of  the columns associated with the vertices of the defective hyperedge $e\in E$. It follows that 
  a non-adaptive group testing strategy successfully detects  the defective hyperedge in $E$ if and only if for any two distinct hyperedges $e,e'\in  E$ we obtain two different response vectors. In terms of the associated binary matrix $M$, this means that the bitwise $OR$ of the columns with indices in $e$ and the $OR$ of the columns with indices in $e'$ are distinct. As for the two-stage algorithm, 
  the non-adaptive algorithm used in the first stage should guarantee to ``separate" at least a certain number of non-defective hyperedges from the defective one.

 \section{Combinatorial structures for group testing in arbitrary hypergraphs}
The following definition provides a combinatorial tool which is essentially equivalent to a non-adaptive group testing algorithm in our search model.
\begin{definition}\label{def:separ}
Given a hypergraph $\cF=(V,E)$ with $V=[n]$ with hyperedges of size at most $d$, we say 
that a $t\times n$ matrix $M$ with entries in $\{0,1\}$ is an  $E$-separable code if for any two distinct hyperedges $e,e'\in E$,  
it holds that $\bigvee_{j\in e} c_j\neq \bigvee_{j\in e'} c_j$, where $c_j$ denotes the column of $M$ with index $j$.
The integer $t$ is the length of the $E$-separable code.
 \end{definition}
Having in mind the correspondence between binary codes and non-adaptive algorithms illustrated in Section \ref{sec:ndgt}, one can see that a non-adaptive algorithm successfully determines the contaminated hyperedge in $E$ if and only if the binary code associated to the algorithm is $E$-separable.
 Therefore, the minimum number of tests of such an algorithm coincides with the minimum length of an $E$-separable code.
 Our existential results for $E$-separable codes are in fact based on existential results  for the  combinatorial structure defined below.
 \begin{definition}\label{def:sel}
Given a $d$-uniform hypergraph $\cF=(V,E)$ with $V=[n]$ and two integers $p$ and $m$ with 
$1\leq p\leq n-d$  and  $1\leq m\leq d+p$, we say 
that a $t\times n$ matrix $M$ with entries in $\{0,1\}$ is an  $(E,p,m)$-selector of size $t$ if for any two distinct hyperedges $e,e'\in E$,  the
submatrix of $M$
consisting of the columns with indices in $S_{e,e',p}$  contains at least m distinct
rows of the identity matrix $I_{d+p}$. 
 \end{definition}
 We will see in Section \ref{sec:nonadaptive} that for $m=d+1$ and $p= \min\{ |e'\setminus e|:\,e,e'\in E \mbox{ and } e\neq e'\}$, an $(E,p,m)$-selector is indeed an $E$-separable code.

The following definition generalizes Definition \ref{def:sel}  and is at the basis of our two-stage algorithm.
 \begin{definition}\label{def:selq}
Given a $d$-uniform hypergraph $\cF=(V,E)$ with $V=[n]$  and integers $q$, $m$ and $\chi$ 
such that $1\leq q\leq |E|-1$,  , and $1\leq m\leq \chi  +d\leq n$,
we say that a $t\times n$ matrix $M$ with entries in $\{0,1\}$ is an  $(E,q,m,\chi)$-selector of size $t$ if, for any $e\in E$ and any other $q$ distinct hyperedges
$e'_1,\ldots,e'_q\in E$, the
submatrix of $M$
consisting of the columns with indices in $S_{e,(e'_1,\ldots,e'_q),\chi}$  contains at least m distinct
rows of the identity matrix $I_{d+\chi}$.
 \end{definition}
For $q=1$ and $\chi=p$, an $(E,q,m,\chi)$-selector is indeed an $(E,p,m)$-selector.  

We remark that even if the combinatorial objects in the present section are defined for uniform hypergraphs, the algorithms we build upon them work also for non-uniform hypergraphs.

 \section{Upper bound on the size of $(E, q,m,\chi)$-selectors }
Let $E$ be a set of hyperedges, each consisting of  $d$ vertices in $[n]$ and, for a given $q$, $1\leq q\leq |E|-1$, let 
 $\chi$ and $m$ be  positive integers with $\chi \leq n-d$ 
and $m\leq d+\chi$.  We will show how to define a hypergraph ${\cal H}=(X,\cB)$, with $X$ being a set of vectors of length $n$, 
in such a way that  the vectors of any {\em cover} of ${\cal H}=(X,\cB)$ are the rows of an $(E, q,m,\chi)$-selector.
We recall that,  given a hypergraph ${\cal H}=(X,\cB)$, a  {\em cover} of $\cal H$ is a subset
$T\subseteq X$ such that for any hyperedge $B\in \cB$ we have
$T\cap B\neq \emptyset$. 
 In order to avoid confusion with the hyperedges of the input hypergraph $E$, here we denote the hyperedges of the above said hypergraph  ${\cal H}=(X,\cB)$ by using the letter $B$.
The following upper bound on the minimum size $\tau(\HC)$ of a cover of a hypergraph  ${\cal H}=(X,\cB)$ is due to  Lov\'asz \cite{lovasz}:
\begin{equation}\label{eq:lovasz}
\tau(\HC)< {|X|\over \min_{B\in  \cB} |B|} (1+\ln\Delta),
\end{equation}
where $\Delta = \max_{x\in X} |\{ \mbox{$B$: $B\in \cB$ and $x\in
B$}\}|$.

In the following, for $e,e'_1,\ldots,e'_q \in E$, we write 
$(i_1,\ldots,i_{d+\chi})=S_{e,(e'_1,\ldots,e'_q),\chi}$ to refer to the tuple of the elements in $S_{e,(e'_1,\ldots,e'_q),\chi}$ arranged in increasing order.
We will show that  $(E,q,m,\chi)$-selectors
are covers of the 
hypergraph 
defined below.
Let $X$ be the set of all binary
vectors $\x=  (x_1,\ldots,x_n)$ of length $n$ con\-taining $\lfloor\frac n{d+\chi}\rfloor$ entries equal to $1$.
For any integer $i$, $1\leq i\leq d+\chi$, let us denote by $\au_i$ the binary
vector of length $d+\chi$ having all components equal to zero but that
in position $i$,  that is,
$\au_1=(1,0,\ldots,0)$, $\au_2=(0,1,\ldots,0)$, \ldots,
$\au_{d+\chi}=(0,0,\ldots,1)$. For any $(q+1)$-tuple $(e,e'_1,\ldots,e'_q)$ of $q+1$ distinct hyperedges in $E$ and, for any binary vector $\au=(a_1,\ldots,a_{d+\chi})\in
\{\au_1,\ldots,\au_{d+\chi}\}$, let us define the set of binary vectors
$B_{\au,{e,(e'_1,\ldots,e'_q)} }=\{\x=(x_1, \ldots , x_n)\in X \,:\,  x_{i_1}=a_1,\ldots, x_{i_{d+\chi}}=a_{d+\chi}, \mbox{ with } (i_1,\ldots,i_{d+\chi})=S_{e,(e'_1,\ldots,e'_q),\chi}\}$.
 For
any set $A\subseteq\{\au_{1},\ldots, \au_{d+\chi}\}$ of size $s$,
$1\leq s\leq d+\chi$,  and any set $S_{e,(e'_1,\ldots,e'_q),\chi}$, for $e,e'_1,\ldots,e'_q \in E$, 
let us define $B_{A,e,(e'_1,\ldots,e'_q)}=\bigcup_{\au\in A} B_{\au,{e,(e'_1,\ldots,e'_q)} }$. For any $s=1,\ldots,d+\chi$, we
define $\cB_s=\{B_{A,{e,(e'_1,\ldots,e'_q)}}\,:\, A \subset\{\au_{1},\ldots, \au_{d+\chi} \},\, |A|=s,  \mbox{ and
  } e,e'_1,\ldots,e'_q \in E \}$ and the hypergraph
 $\HC_s=(X,\cB_s)$.
 We claim that {\em any} cover $T$ of $\HC_{d+\chi-m+1}$ is an $(E,q,m,\chi)$-selector, that is, for any $e,e'_1,\ldots,e'_q \in E$, the submatrix
 of $d+\chi$ columns of $T$ with indices in $S_{e,(e'_1,\ldots,e'_q),\chi}$
contains at least $m$ distinct rows of the identity matrix $I_{d+\chi}$.
The proof is by contradiction. Assume that there exist $q+1$ hyperedges $e,e'_1,\ldots,e'_q\in E$ such that the submatrix  of  the columns of $T$  with indices in the $(d+\chi)$-tuple $(i_1,\ldots,i_{d+\chi})=S_{e,(e'_1,\ldots,e'_q),\chi}$ contains {\em at most}
$m-1$ distinct rows of $I_{d+\chi}$. Let such  rows be
$\au_{j_1}, \ldots , \au_{j_g}$, with $g\leq m-1$, and
let $A$ be {\em any} subset of $\{\au_1, \ldots , \au_{d+\chi}\}\setminus \{\au_{j_1}, \ldots , \au_{j_g}\}$
of cardina\-lity $|A|=d+\chi-m+1$. By definition of  $\HC_{d+\chi-m+1}$,  the  hyperedge  $B_{A, e,(e'_1,\ldots,e'_q)}$  of $\HC_{d+\chi-m+1}$ is such that $T\cap B_{A,e,(e'_1,\ldots,e'_q)}=\emptyset$, thus contradicting the fact that $T$ is a cover for
$\HC_{d+\chi-m+1}$.

\smallskip
To prove the following theorem, we 
exploit Lov\'asz's result (\ref{eq:lovasz})
to derive an upper bound on the minimum size of a cover of the hypergraph $\HC_{d+\chi-m+1}$.

\begin{theorem}\label{thm:lovasz_sel}
Let $E$ be a set of hyperedges of size $d$ with vertices in $[n]$. Moreover, let $q$, $m$ and $\chi$ be integers 
such that    $1\leq q\leq |E|-1$  and $1\leq m\leq \chi +d\leq n$.
There exists an  
$(E,q,m,\chi)$-selector of size $t$ with 
\begin{equation}\label{upperselector}
t<\frac {2e(d+\chi)}{d+\chi-m+1} \left(1+ \ln\left({d+\chi-1\choose d+\chi-m}\alpha\right)\right),
\end{equation}
where $\alpha=\min\left\{e^q|E|\left({|E|-1\over q}\right)^q, e^{d+\chi-1} \left({n\over d+\chi-1}\right)^{d+\chi}\right\}$ and $e=2.7182...$ is the base of the natural logarithm.
\end{theorem}
\begin{proof} 
Let $\HC_{d+\chi-m+1}=(X,\cB_{d+\chi-m+1})$ be the hypergraph defined in the discussion preceding the statement of the theorem. Inequality   (\ref{eq:lovasz})  implies that 
\begin{equation}\label{eq:lovasz2}
\tau(\HC)< {|X|\over  \min \{|B|\; : B\in \cB_{d+\chi-m+1}\}} (1+\ln\Delta).
\end{equation}
In order to derive an upper bound on $t$, we estimate the quantities that appear on the right hand side of  (\ref{eq:lovasz2}), that is,
we evaluate the quantities
$$|X|, \ \ \ \min \{|B|\; : B\in \cB_{d+\chi-m+1}\}, \ \ \ \hbox{ and } \ \Delta,$$ for the hypergraph
$\HC_{d+\chi-m+1}$. 
Since $X$ is the set of all binary
vectors  of length $n$ containing $\lfloor\frac n{d+\chi}\rfloor$ entries equal to $1$,
 it holds that $|X|= {{n} \choose{  \lfloor\frac n{d+\chi}\rfloor}}$.
Each $B_{\au, e,(e'_1,\ldots,e'_q)}$ has size ${n-d-\chi\choose \lfloor\frac n{d+\chi}\rfloor-1}$ since this is the number of vectors, with $\lfloor\frac n{d+\chi}\rfloor$ entries equal to $1$, for which one has  that the entries with indices in $(i_1,\ldots,i_{d+\chi})=S_{e,(e'_1,\ldots,e'_q),\chi}$  form the vector $\au$.  
Moreover, each hyperedge $B_{A,e,(e'_1,\ldots,e'_q)}$ of $\HC_{d+\chi-m+1}$ is the union of
$d+\chi-m+1$ disjoint sets $B_{\au, e,(e'_1,\ldots,e'_q)}$, and therefore it has cardinality
$$|B_{A,e,(e'_1,\ldots,e'_q)}|= (d+\chi-m+1)\cdot |B_{\au,e,(e'_1,\ldots,e'_q)}|=(d+\chi-m+1){n-d-\chi\choose \lfloor\frac n{d+\chi}\rfloor-1},$$
and consequently, one has that 
$$\frac{ |X|} {\min \{|B|\; : B\in \cB_{d+\chi-m+1}\}}= {{n\choose \lfloor\frac n{d+\chi}\rfloor}\over  {(d+\chi-m+1){n-d-\chi\choose \lfloor\frac n{d+\chi}\rfloor-1}}}.$$

Let us prove an upper bound 
on ${{n\choose \lfloor\frac n{d+\chi}\rfloor}\over  {(d+\chi-m+1){n-d-\chi\choose \lfloor\frac n{d+\chi}\rfloor-1}}}$.
For $d+\chi>n/2$, it is
 ${{n\choose \lfloor\frac n{d+\chi}\rfloor}\over  {n-d-\chi\choose \lfloor\frac n{d+\chi}\rfloor-1}}=n< 2(d+\chi)$, whereas for $d+\chi\leq n/2$ it is
\begin{eqnarray*}
{{n\choose \lfloor\frac n{d+\chi}\rfloor}\over  {n-d-\chi\choose \lfloor\frac n{d+\chi}\rfloor-1}}&=&
\frac n{\lfloor\frac n{d+\chi}\rfloor}\cdot\frac{n-1}{n-\lfloor\frac n{d+\chi}\rfloor}\cdot \frac{n-2}{n-\lfloor\frac n{d+\chi}\rfloor-1}\cdot\ldots\cdot\frac{n-d-\chi+1}{n-d-\chi-\lfloor\frac n{d+\chi}\rfloor+2}\cr   
& &\hphantom{aa}\cr
&\leq&
(d+\chi)\cdot\frac n{n-d-\chi}\cdot\frac{n-1}{n-\lfloor\frac n{d+\chi}\rfloor}\cdot \frac{n-2}{n-\lfloor\frac n{d+\chi}\rfloor-1}\cdot\ldots\cdot\frac{n-d-\chi+1}{n-d-\chi-\lfloor\frac n{d+\chi}\rfloor+2}\cr   
& &\hphantom{aa}\cr
&\leq&
2(d+\chi)\cdot \frac{n-1}{n-\lfloor\frac n{d+\chi}\rfloor}\cdot \frac{n-2}{n-\lfloor\frac n{d+\chi}\rfloor-1}\cdot\ldots\cdot\frac{n-d-\chi+1}{n-d-\chi-\lfloor\frac n{d+\chi}\rfloor+2}\cr
&\leq& 2(d+\chi) \left(\frac{n-d-\chi+1}{n-d-\chi-\lfloor\frac n{d+\chi}\rfloor+2}\right)^{d+\chi-1}\cr
& &\hphantom{aa}\cr
&\leq& 2(d+\chi) \left(\frac{n-d-\chi+1}{n-d-\chi-\frac n{d+\chi}+2}\right)^{d+\chi-1}\cr
& &\hphantom{aa}\cr
&=&2(d+\chi)\left({{(d+\chi)(n-d-\chi+1)}\over{(d+\chi)( n-d-\chi+1)-(n-d-\chi)}}\right)^{d+\chi-1}\cr
& &\hphantom{aa}\cr
&=&2(d+\chi)\left(1+ {{n-d-\chi}\over {(d+\chi)(n-d-\chi+1)-(n-d-\chi)}}\right)^{d+\chi-1}\cr
& &\hphantom{aa}\cr
&\leq&2( d+\chi )\left(1+ {1\over{d+\chi-1}}\right)^{d+\chi-1}\cr
& \hphantom{aa}\cr
 &\leq& 2e(d+\chi).
 \end{eqnarray*}
 Therefore, we have that
\begin{equation}\label{eq:ratio}
\frac{ |X|} {\min \{|B|\; : B\in \cB_{d+\chi-m+1}\}}\leq\frac {2e(d+\chi)}{d+\chi-m+1}.
\end{equation}

In order to compute $\Delta$, we notice that each $\x\in X$ belongs to at most
\begin{equation}\label{eq:fordelta}
\min\left\{|E|{|E|-1\choose q},{ \lfloor\frac n{d+\chi}\rfloor \choose 1}{{n-\lfloor\frac n{d+\chi}\rfloor}\choose {d+\chi-1}}\right\}
\end{equation}
distinct sets $B_{\au, e,(e'_1,\ldots,e'_q)}$.
Indeed, as observed in Section \ref{sec:notation}, the number of distinct sets $S_{e,(e'_1,\ldots,e'_q),\chi}$ is at most $|E|{|E|-1\choose q}$.  Notice that, for a vector $\au=(a_1,\ldots,a_{d+\chi})$, the actual number of $(d+\chi)$-tuples $S_{e,(e'_1,\ldots,e'_q),\chi}=(i_1,\ldots,i_{d+\chi})$ for which 
it holds that  $x_{i_1}=a_1,\ldots, x_{i_{d+\chi}}=a_{d+\chi}$ might be significantly smaller than the total number of $(d+\chi)$-tuples $S_{e,(e'_1,\ldots,e'_q),\chi}$, but this estimate is sufficient to obtain the first value  in the 
``min'' in the expression of $\alpha$ that appears in the statement of the theorem. 
However, if $|E|$ is close to ${n\choose d}$, i.e., it contains almost all hyperedges of size $d$ on $[n]$,
then it might be  convenient to upper bound the number of distinct hyperedges $B_{\au, e,(e'_1,\ldots,e'_q)}$ containing a given 
vector $\x\in X$ by 
${ \lfloor\frac n{d+\chi}\rfloor \choose 1}{{n-\lfloor\frac n{d+\chi}\rfloor}\choose {d+\chi-1}}$ which is obtained by considering all possible ways of choosing $(d+\chi)$ positions in $\x$ so that exactly one of those positions corresponds to a 1-entry of $\x$, whereas the others correspond to 0-entries. 

Now observe that 
each $B_{\au, e,(e'_1,\ldots,e'_q)}$ belongs to ${{d+\chi-1}\choose {d+\chi-m}}$ distinct hyperedges
$B_{A,e,(e'_1,\ldots,e'_q)}$. Therefore, (\ref{eq:fordelta}) implies that the maximum degree of vertices in  $\HC_{d+\chi-m+1}$ is

\begin{equation}\label{eq:delta1}
\Delta\leq {d+\chi-1\choose d+\chi-m}\min\left\{|E|{|E|-1\choose q},{ \lfloor\frac n{d+\chi}\rfloor \choose 1}{{n-\lfloor\frac n{d+\chi}\rfloor}\choose {d+\chi-1}}\right\}.
\end{equation}

In order to estimate our upper bound, we  resort to the following well  known inequality:
 \begin{equation}\label{eq:binom}
 {a\choose b}\leq (ea/b)^b.
 \end{equation}
By exploiting (\ref{eq:binom}), we obtain the following upper bound of the second expression in the min of (\ref{eq:delta1}):
\begin{eqnarray*}
{ \lfloor\frac n{d+\chi}\rfloor  \choose 1} {n- \lfloor\frac n{d+\chi}\rfloor  \choose d+\chi-1}
 &\le&
 \left\lfloor\frac n{d+\chi}\right\rfloor  \left({n- \lfloor\frac n{d+\chi}\rfloor )\over d+\chi-1}\right)^{d+\chi-1}e^{d+\chi-1}\cr\cr
  &\leq & e^{d+\chi-1}\left({n\over d+\chi-1}\right)^{d+\chi}.
\end{eqnarray*}
Applying inequality  (\ref{eq:binom}) to ${|E|-1\choose q}$ in (\ref{eq:delta1}), we get
\remove{
Applying inequality  (\ref{eq:binom}) to ${|E|-1\choose q}$ and to ${{n-\lfloor\frac n{d+\chi}\rfloor}\choose {d+\chi-1}}$ in the upper bound (\ref{eq:delta1}) on $\Delta$, we get that }
\begin{equation}\label{eq:delta3}
\!\!\!\!\!\!\Delta\leq {d+\chi-1\choose d+\chi-m}\min\left\{e^q|E|\left({|E|-1\over q}\right)^q, e^{d+\chi-1} \left({n\over d+\chi-1}\right)^{d+\chi}\right\}
\end{equation}
\remove{Now, let us limit from above $\frac{|X| } { \min \{|B|\; : B\in \cB_{d+\chi-m+1}\}}= {{n\choose {n\over d+\chi}}\over  {(d+\chi-m+1){n-d-\chi\choose \lfloor\frac n{d+\chi}\rfloor-1}}}$.
For $d+\chi\in\{1,2\}$, it is
 ${{n\choose \lfloor\frac n{d+\chi}\rfloor}\over  {n-d-\chi\choose \lfloor\frac n{d+\chi}\rfloor-1}}< 2(d+\chi)$, whereas for $d+\chi\geq 3$ it is
\begin{eqnarray*}
{{n\choose \lfloor\frac n{d+\chi}\rfloor}\over  {n-d-\chi\choose \lfloor\frac n{d+\chi}\rfloor-1}}&=&
(d+\chi) \frac{n-1}{n-\lfloor\frac n{d+\chi}\rfloor}\cdot \frac{n-2}{n-\lfloor\frac n{d+\chi}\rfloor-1}\cdot\ldots\cdot\frac{n-d-\chi+1}{n-d-\chi-\lfloor\frac n{d+\chi}\rfloor+2}\cr
&\leq& (d+\chi) \left(\frac{n-d-\chi+1}{n-d-\chi-\lfloor\frac n{d+\chi}\rfloor+2}\right)^{d+\chi-1}\cr
& &\hphantom{aa}\cr
&=&(d+\chi)\left({{(d+\chi)(n-d-\chi+1)}\over{(d+\chi)( n-d-\chi+1)-(n-d-\chi)}}\right)^{d+\chi-1}\cr
& &\hphantom{aa}\cr
&=&(d+\chi)\left(1+ {{n-d-\chi}\over {(d+\chi)(n-d-\chi+1)-(n-d-\chi)}}\right)^{d+\chi-1}\cr
& &\hphantom{aa}\cr
&\leq&( d+\chi )\left(1+ {1\over{d+\chi-1}}\right)^{d+\chi-1}
 <   e(d+\chi).
\end{eqnarray*}
The above inequalities imply}
By plugging into (\ref{eq:lovasz2}) the upper bound (\ref{eq:delta3}) on $\Delta$ and the upper bound (\ref{eq:ratio}) on $\frac{ |X|} {\min \{|B|\; : B\in \cB_{d+\chi-m+1}\}}$, we obtain
the stated upper bound. 
\end{proof}
\subsection{A non-adaptive group testing algorithm for general hypergraph}\label{sec:nonadaptive}
In order to prove the upper bound on the number of tests of the non-adaptive algorithm, we prove a more general result that can be exploited also to prove existential results for the three-stage algorithms.
\begin{theorem}\label{thm:gen_non_adaptgt}
Let $\cF=(V,E)$  be a hypergraph with $V=[n]$ with all hyperedges in $E$ of size {\em at most} $d$. For any positive integer $p\leq d-1$, there exists a  non-adaptive algorithm that allows to discard all but those hyperedges $e$ such that $|e\setminus e^*|< p$, where $e^*$ is the defective hyperedge, and uses
$t=O\left(\frac d p\log E\right)$ tests.
\end{theorem}
\begin{proof}
Let $E$ denote the input set of hyperedges. First we consider the case  when all hyperedges of $E$  have size exactly $d$ and show that in this case an $(E,p,,d+1)$-selector corresponds to a non-adaptive algorithm that allows to discard all but those hyperedges  $e\in E$ such that $|e\setminus e^*|< p$. Then we will consider the case of  hypergraphs that are not necessarily uniform.
Let $M$ be an $(E,p,d+1)$-selector. By  Definition \ref{def:sel},
for any two distinct hyperedges $e$ and $e'$, the submatrix $M'$ of $M$ formed by the $d+p$ columns with indices in $S_{e,e',p}$ contains at least $d+1$ rows of the identity matrix $I_{d+p}$. 
 At least one of these rows has 0 at the intersection with each of the $d$ columns 
 with index in $e$ and an entry equal to 1 at the intersection with one of the $p$ columns with indices in $S_{e,e',p}\setminus e$. 
 By definition of $S_{e,e',p}$, if $|e'\setminus e|\geq p$ then $S_{e,e',p}\setminus e\subseteq e'\setminus e$,
  and consequently the above said entry equal to 1 is at the intersection with one column with index in $e'\setminus e$. It follows that the OR of the  $d$ columns  associated to vertices in $e$ is 
different from that of the  $d$ columns  associated to vertices in $e'$. 
Now suppose that  $e^*$ be the defective hyperedge and $e'$ be an arbitrary hyperedge such that $|e'\setminus e^*|\geq p$.
From the above argument, there exists a row index, say $r$, such that at least one column associated to a vertex in $e'$ has an entry equal to 1  in position $r$, whereas all columns associated to vertices in $e^*$ have 0 in that position. 
It follows that  the response vector has an entry equal to 0 in position $r$ whereas at least one element of $e'\setminus e^*$ is in the pool associated to the row of $M$ with index $r$, thus showing that $e'$ is not the defective hyperedge. By inspecting each column of the $(E,p,d+1)$-selector and comparing it with the response vector, one can identify 
those columns that have an entry equal to 1 in a position where the response vector has 0 and get rid of all hyperedges containing one of the vertices associated to those columns. By the above argument one gets rid of all hyperedges $e'$ such that $|e'\setminus e^*|\geq p$. 

 Now let us consider a not necessarily uniform hypergraph $E$  with hyperedges of size at most $d$. We  add dummy vertices to each hyperedge of size smaller than $d$, so as to obtain a set $\tilde E$ of hyperedges all having size $d$. 
 We will see that a non-adaptive group testing based on an $(\tilde E,p,d+1)$-selector allows to discard all non-defective hyperedges $e$ such that
 $|e\setminus e^*|\geq p$.
  Let $\{i_{n+1},\ldots,i_{n+d-1}\}$ be the dummy vertices. Notice that these dummy vertices do not need to be different from those we have used in the definition of
   the sets $S_{e,e',p}$'s, as we will see below.  For each $e\in E$, we denote with $\tilde e$ the corresponding hyperedge 
 in $\tilde E$. The hyperedge $\tilde e$  is either  equal to $e$, if $|e|=d$,  or equal to $e\cup\{i_{n+1},\ldots,i_{n+d-|e|}\}$, if $|e|<d$.
  Obviously, $\tilde E$ has all hyperedges of size $d$ and it holds that $|\tilde E|=|E|$. 
 Let $M$ be an $(\tilde E,p,d+1)$-selector.   
By Definition \ref{def:sel},  for any two distinct hyperedges $\tilde e$ and $\tilde e'$, the submatrix $M'$ of $M$ formed by the $d+p$ columns with indices in $S_{\tilde e,\tilde e',p}$ contains at least $d+1$ rows of the identity matrix $I_{d+p}$. 
 At least one of these rows, say the row with index $r$, has 0 at the intersection with all the columns  
 with index in $\tilde e$ and an entry equal to 1 at the intersection with one of the remaining columns of $M'$. 
 If $|\tilde e'\setminus \tilde e|\geq p$, it holds that  $S_{\tilde e,\tilde e',p}=S_{\tilde e,\tilde e}$. In other words, all column indices in $S_{\tilde e,\tilde e',p}$ 
  that are not in $\tilde e$, are in fact all contained in $\tilde e'\setminus \tilde e$ and consequently the above said entry equal to 1 in the row with index $r$ is at the intersection with one of the columns in 
  $\tilde e'\setminus \tilde e$. 
Now, let $e^*$ denote the defective hyperedge in the original set of hyperedges $E$ and let us replace $e$ with $e^*$ in the above discussion.
 Suppose that a hyperedge $e'$ of the original hypergraph  is such that $|e'\setminus e^*|\geq p$.  Since for each $e\in E$ it holds that $e\subseteq \tilde e$ and 
 $e\cap \{i_{n+1},\ldots i_{n+\chi}\}=\emptyset$, it follows that $\tilde e\setminus \tilde e'\supseteq e\setminus e'$. 
 This implies that $|\tilde e'\setminus \tilde e^*|$ is larger than or equal to $p$, and consequently, 
 from the above discussion, it follows that the submatrix $M'$ formed by the columns of $M$ with index in $S_{\tilde e,\tilde e'}$ contains  a row that has 0 at the intersection with each of the columns  
 with index in $\tilde e^*$ and an entry equal to 1 at the intersection with one of  columns with index in  $\tilde e'\setminus\tilde e^*$.
  We need to show that this index does not correspond to a dummy vertex, i.e., we need to show that it belongs to  $e'\setminus  e^*$.
  To  see this, we recall that 
$S_{\tilde e^*, \tilde e}\cap( \tilde e\setminus \tilde e^*)$ consists of the smallest $p$ vertices in 
$\tilde e'\setminus \tilde e^*$ and since $|e'\setminus e^*|\geq p$, it holds that these $p$ vertices belong all to $e'\setminus e^*$.
Therefore, we have proved that there exists a row in $M'$, that contains 0 at the intersection with all the columns with index in $\tilde e^*$ and  an entry equal to 1
 at the intersection with a column associated to a vertex in $e'\setminus  e^*$.  Since $\tilde e^*\supseteq e^*$, this row has all 0's  at the intersection 
 with  the columns with index in $e^*$. As a consequence, for any non-defective
hyperedge $e'$ with $|e'\setminus e^*|\geq p$, there exists a row index $r$ such that at least one column with index  in $e'$ has a  1-entry in position $r$,  
whereas the response vector has a 0-entry in that position. In other words there is an element  $x$ in $e'$ that is contained in a pool that tests negative, and consequently, $e'$
is not the defective hyperedge. 

Notice that decoding can be achieved by inspecting all columns of the selector that are associated with non-dummy vertices.  Indeed, as we have just seen, any non-defective hyperedge $e'$ such that $|e'\setminus e^*|\geq p$ contains a vertex that is included in a pool that tests negative, i.e,  a vertex associated with a column with a 1 in a position that corresponds  to a 0 in the response vector. By inspecting all columns of the selectors, the decoding algorithm finds all columns with a 1 in a position corresponding to a negative response and gets rid of all hyperedges that contain one or more vertices associated with those columns. In this way, the decoding algorithm discards any non-defective hyperedge
 $e'$ such that $|e'\setminus e^*|\geq p$. 

Finally, we  observe that the  dummy vertices  added to the hyperedges of $E$ do not need to be different from those used in the definition of
   the sets $S_{e,e',p}$ in Section \ref{sec:notation}.  
  The fact that some of the dummy vertices used to define the sets $S_{e,e',p}$'s might be contained in some of the hyperedges
   of $\tilde E$ might lead
    to a fault  in the above discussion only  if $|\tilde e'\setminus \tilde e|<p $ and we need to include dummy vertices in $S_{\tilde e,\tilde e',p}$. 
    Indeed, some dummy vertices might
   be also included in $\tilde e$ and we could end up at adding some vertices twice to the set $S_{\tilde e,\tilde e',p}$. 
   However, this cannot happen since in the above discussion we have exploited the property of $(\tilde E,p,d+1)$-selectors 
only in reference to pairs of hyperedges $\tilde e$ and $\tilde e'$ associated to vertices $e$ and $e'$
   of $E$ such that $|e\setminus e'|\geq p$. 
   We have seen that this implies not only
 that $S_{\tilde e,\tilde e',p}$ is equal to
   $S_{\tilde e,\tilde e'}$ but also that $S_{\tilde e,\tilde e'}\cap (\tilde e'\setminus \tilde e)\subseteq e'\setminus e$. In other words, the only dummy vertices possibly involved in the definition of $S_{\tilde e,\tilde e',p}$ are those in $\tilde e$.   
   
   The upper bound in the statement of the theorem then follows from the upper bound Theorem \ref{thm:lovasz_sel}  with $\chi=p$, $q=1$, and $m=d+1$. Notice that we have replaced $n$ with $n+d-1$  in that bound, since we have added up to $d-1$ dummy vertices to $[n]$ in order to obtain the set of hyperedges $\tilde E$.
\end{proof}

The following corollary is an immediate consequence of Theorem \ref{thm:gen_non_adaptgt}.
\begin{corollary}\label{cor:non-adaptivegt}
Let $d$ and $n$ be integers with $1\leq d\leq n$, and let $E$ be  a set of hyperedges of size {\em at most} 
$d$ on $[n]$. Moreover, let $p$  be an integer such that  
$1\leq p\leq min\{ |e'\setminus e|:\,e,e'\in E\}$.
There exists a non-adaptive group testing algorithm that finds the defective hyperedge in $E$  and uses at most 
$t=O\left( \frac dp\log |E|\right)$
tests. 
\end{corollary}
\begin{proof}
If  for any two distinct hyperedges $e,e'\in E$ it holds that 
$min\{ |e'\setminus e|:\,e,e'\in E\}\geq p$, then the non-adaptive algorithm of Theorem \ref{thm:gen_non_adaptgt} discards all non-defective hyperedges.
\end{proof} 
\remove{
\begin{corollary}\label{cor:non-adaptivegt}
Let $d$ and $n$ be integers with $1\leq d\leq n$, and let $E$ be  a set of hyperedges of size $d$ on $[n]$. Moreover, let $p$  be an integer such that  $
1\leq p\leq min\{ |e'\setminus e|:\,e,e'\in E\}$. If $p=\Theta(d)$ then there exists a non-adaptive group testing algorithm that finds the defective hyperedge in $E$  and uses at most 
$O\left(\log |E|\right)$
tests.
\end{corollary}
}
\subsection{A two-stage group testing algorithm for general hypergraphs}
\begin{theorem}\label{thm:two-stagegt}
Let $\cF=(V,E)$  be a hypergraph with $V=[n]$  and all hyperedges in $E$ of size at most $d$. Moreover, let $q$ and $\chi$ be positive integers 
such that    $1\leq q\leq |E|-1$, $\chi =\min\{ |\bigcup_{i=1}^q e'_i\setminus e|,  \mbox{ for any $q+1$ distinct $e,e'_1,\ldots,e'_q\in E$}\}$.
There exists a  (trivial) two-stage algorithm that uses  a number of tests $t$ with
$$
t< \frac {2e(d+\chi)}{\chi}  \left(1+ \ln\left({d+\chi-1\choose d+\chi-d-1}\beta\right)\right)+dq,$$
where $\beta=\min\left\{e^q|E|\left({|E|-1\over q}\right)^q, e^{d+\chi-1}\left({n+d-1\over d+\chi-1}\right)^{d+\chi}\right\}$ and $e=2.7182...$ is the base of the natural logarithm.
\end{theorem}

\begin{proof}
 Let $E$ be a set of hyperedges each consisting of at most $d$ vertices in $[n]$.
 As in the proof of Corollary \ref{cor:non-adaptivegt}, let us denote with $\tilde E$ the set of hyperedges obtained by 
 replacing in $E$ any hyperedge $e$ of size smaller than $d$  by $e\cup\{i_{n+1},\ldots,i_{n+d-|e|}\}$. 
 For each $e\in E$, we denote with $\tilde e$ the corresponding hyperedge 
 in $\tilde E$. This hyperedge is either  equal to $e$, if $|e|=d$,  or is equal to $e\cup\{i_{n+1},\ldots,i_{n+d-|e|}\}$, if $|e|<d$. Let us consider the non-adaptive algorithm that tests the pools having as characteristic vectors the rows of an $(\tilde E,q,d+1,\chi)$-selector $M$. We will show that after executing this non-adaptive algorithm, one is left with at most $q$ hyperedges that are candidate to be the defective hyperedge.
 Let $e^*$ be the defective hyperedge and let us suppose by contradiction that, after executing the algorithm associated with $M$, there are at least $q$ hyperedges,
 in addition to $e^*$, 
that are still candidate to be the defective one.  Let us consider $q$ such hyperedges, say $e'_1,\ldots,e'_{q}$. 

By Definition \ref{def:selq}, it holds that the submatrix $M'$ of $M$
consisting of the columns with indices in $S_{\tilde e^*,(\tilde  e'_1,\ldots,\tilde e'_q),\chi}$  contains at least $d+1$ distinct
rows of the identity matrix $I_{d+\chi}$. This implies that at least one of these rows, say the row with index $r$, has  a 0  at the intersection with each column
 with index in $\tilde e^*$ and   1 at the intersection with one of the remaining columns of $S_{\tilde e^*,(\tilde  e'_1,\ldots,\tilde e'_q),\chi}$.
 Since for each $e\in E$ it holds that $e\subseteq \tilde e$ and 
 $e\cap \{i_{n+1},\ldots i_{n+\chi}\}=\emptyset$, it follows that
 $\bigcup_{i=1}^q \tilde e'_i\setminus \tilde e^*\supseteq \bigcup_{i=1}^q e'_i\setminus  e^*$. Consequently, the hypothesis  $|\bigcup_{i=1}^q e'_i\setminus  e^*|\geq \chi$ implies that  $|\bigcup_{i=1}^q \tilde e'_i\setminus \tilde e^*|\geq\chi$.
  This means that  $S_{\tilde e^*,(\tilde  e'_1,\ldots,\tilde e'_q),\chi}=S_{\tilde e^*,(\tilde  e'_1,\ldots,\tilde e'_q)}$, i.e., that the columns with index in $S_{\tilde e^*,(\tilde  e'_1,\ldots,\tilde e'_q),\chi}$ that are not in $\tilde e^*$, are in fact all contained in  $\bigcup_{i=1}^q \tilde e'_i\setminus \tilde e^*$. Therefore, the above said entry equal to
  1 in the row with index $r$ is at 
 the intersection with  one of the columns in 
 $\bigcup_{i=1}^q \tilde e'_i\setminus \tilde e^*$. 
 Now we need to prove that  this entry equal to 1,  in fact, intersects a column associated with a non-dummy vertex, i.e., a vertex in $\bigcup_{i=1}^q e'_i\setminus  e^*$.  By definition of $S_{\tilde e^*,(\tilde  e'_1,\ldots,\tilde e'_q),\chi}$ , the set $S_{\tilde e^*,(\tilde  e'_1,\ldots,\tilde e'_q),\chi}$ 
consists of the smallest $\chi$ integers  in $\bigcup_{i=1}^q \tilde e'_i\setminus \tilde e^*$. Moreover, since  $\bigcup_{i=1}^q  e'_i\setminus  e^*
\subseteq \bigcup_{i=1}^q \tilde e'_i\setminus \tilde e^*$ and  $|\bigcup_{i=1}^q  e'_i\setminus  e^*|\geq \chi$, one has that
$S_{\tilde e^*,(\tilde  e'_1,\ldots,\tilde e'_q),\chi}\cap ( \bigcup_{i=1}^q \tilde e'_i\setminus \tilde e^*)\subseteq   \bigcup_{i=1}^q e'_i\setminus  e^*$, thus implying that 
the above said entry equal to 1 of the row with index $r$ is at the intersection with a column associated to a vertex in $\bigcup_{i=1}^q  e'_i\setminus  e^*$.  
We also observe that since $e^*\subseteq \tilde e^*$, all columns in $e^*$ have zeros  at the intersection with the row with index $r$. Let $T_r$ be the pool having this row as characteristic vector. 
 From what we have just said, it holds that $|T_r\cap e^*|=0$ and $|T\cap \bigcup_{i=1}^q e'_i\setminus e^*|\geq1$. As a consequence,  $T_r$ contains at least one element 
 belonging to one of the hyperedges $e'_1,\ldots,e'_q$ whereas the result of the test on $T_r$ is negative,   thus indicating that at least one of  $e'_1,\ldots,e'_q$ is non-defective.
 
 Let us consider now a two-stage algorithm whose first  stage consists in the non-adaptive group testing algorithm
 associated with the rows of the $(\tilde E,q,d+1,\chi)$-selector $M$. From the above argument, after the first stage, we are left with at most $q$ hyperedges that
 are still candidate to be the defective hyperedge. Therefore, in order to determine the defective hyperedge, one needs only to perform individual tests on at most $dq$ elements. 
 Notice that these tests can be performed in parallel. 
 
 As for the  decoding algorithm that identifies the non-defective hyperedges from the responses to the tests performed in stage 1,  the algorithm needs only to inspect each column of the selector and every time 
 finds a column with a 1-entry in a position corresponding to a no response, discards  the hyperedges that contain the vertex associated with that column.  
 The stated upper bound follows by the upper bound of Theorem \ref{thm:lovasz_sel} by replacing $m$ with $d+1$ and $n$ with $n+d-1$ since the set of vertices of $E$ has been augmented with at most 
 $d-1$ dummy vertices.
 \end{proof}
\begin{corollary}
Let $d$ and $n$ be integers with $1\leq d\leq n$, and let $E$ be  a set of hyperedges of size at most $d$ on $[n]$. If there exists   a constant $q\geq 1$ such that for any $q+1$ distinct hyperedges $e,e'_1,\ldots,e'_q\in E$,
$ |\bigcup_{i=1}^q e'_i\setminus e|= \Omega(d)$, then there exists a trivial two-stage algorithm  that finds the defective hyperedge in $E$  and uses
$t=O(\log |E|)$ tests. 
\end{corollary}
\subsection{A three-stage group testing algorithm for general hypergraphs}
The following theorem furnishes an upper bound on the number of tests of three-stage algorithms. An interesting feature of this upper bound is that it holds independently of the size of the pairwise intersections of 
the hyperedges. 
\begin{theorem}\label{thm:three-stagegtgen}
Let $\cF=(V,E)$  be a hypergraph with $V=[n]$  with hyperedges of size {\em at most} $d$ and let $b$ be any positive integer smaller than $d$. There exists a  three-stage algorithm that finds the defective hyperedge in $E$ and uses $t=O(\frac db \log |E|+b\log |E|)$ tests. 
\end{theorem}
\begin{proof}
Let us describe the three stage algorithm. The first two stages aim at restricting the set of potentially defective hyperedges to  hyperedges with at most $b-1$ vertices. 
From Theorem \ref{thm:gen_non_adaptgt}, one has that there is an  $O(\frac db \log |E|)$ 
non-adaptive algorithm $\cal A$ that discards all but the hyperedges $e$ such that $|e\setminus e^*|<  b$, where $e^*$ is the defective hyperedge. Stage 1  consists in running algorithm $\cal A$. If all hyperedges that have not be discarded by $\cal A$ have size smaller than $b$, then the algorithm skips  stage 2 and proceeds to  stage 3.  If otherwise, the algorithm chooses a hyperedge $e$ of maximum size among those that have not been discarded by  $\cal A$ and 
 proceeds to   stage 2 where it  performs in parallel individual tests on the vertices of $e$. Since $e$ has not been discarded, it means that $|e\setminus e^*|<  b$ and therefore at least $|e|-b+1$ of the individual tests yield a positive response. Let $i_1,\ldots,i_f$ denote the vertices of $e$ that have been tested positive. 
The algorithm looks at the intersection between $e$ and any other not yet discarded hyperedge $e'$ and discards $e'$ if it either contains one or more of the vertices of $e$ that have
been tested negative, or if   $\{i_1,\ldots,i_f\}\not\subseteq e'$. In other words, after this second stage the algorithm is left only
with the hyperedges $e'$ such that  $e\cap e'=\{i_1,\ldots,i_f\}$. The vertices $i_1,\ldots,i_f$ are removed from all these hyperedges since it is already known that they are defective.
Each hyperedge $e'$ that has not been discarded in stage 2 is therefore replaced by a hyperedge of size  $|e'| -f\leq |e'|-|e|+b-1\leq b-1$. The last inequality is a consequence of having chosen  $e$ as a hyperedge of maximum size among  those that have not been discarded in stage 1. 
In stage 3, the algorithm is left with a hypergraph with hyperedges of at most $b-1$ vertices and, by Corollary \ref{cor:non-adaptivegt} with $d$ being replaced by $b-1$, the defective hyperedge can be detected non-adaptively using $O(b\log |E|)$ tests. In applying  Corollary \ref{cor:non-adaptivegt}, we do not make any assumption on the size of the set differences between  hyperedges and take $p$ in that bound as small as 1. Notice that the algorithm of  Corollary \ref{cor:non-adaptivegt} might end up with more that one hyperedge if  there exist hyperedges that are proper subsets of the defective one. However, in this eventuality the algorithm would choose the largest hyperedge among those that have not be discarded.
 \end{proof}
By setting $b=\sqrt d$ in Theorem \ref{thm:three-stagegtgen}, we get the following corollary.
\begin{corollary}\label{cor:two-stagegtsqrt}
Let $E$  be a hypergraph with $V=[n]$  with hyperedges  of size at most $d$. There exists a three-stage algorithm that finds the defective hyperedge in $E$ and uses $t=O({\sqrt d}\log |E|)$ tests. 
\end{corollary}


\begin{thebibliography}{5}

\message{References}
\bibitem{alon} Alon, N.,  Asodi, V.: Learning a hidden subgraph.   SIAM J. Discrete Math. 18, no. 4, pp. 697--712  (2005)

\bibitem{arasli} Arasli, B. and Ulukus, S.: Graph and cluster formation based group testing, 2021 IEEE ISIT, pp. 1236-1241 (2021)




\bibitem{cms}  Clementi, A. E. F., Monti, A.,  Silvestri, R.: Selective families, superimposed codes, and
broadcasting on unknown radio networks. In:  Twelfth Annual ACM-SIAM Symposium on Discrete Algorithms, pp. 709--718 (2001)



\bibitem{siam} De Bonis, A,  G\c asieniec, L, Vaccaro, U.: Optimal two-stage algorithms for group testing problems.   SIAM J. Comput.   34, no. 5, pp.1253--1270 (2005)

\bibitem{dor} Dorfman, R.: The detection of defective members of large populations.   Ann. Math. Statist. 14, pp. 436--440 (1943)
\bibitem{dh1} Du, D.Z.,  Hwang, F. K.:   Pooling design and Nonadaptive Group Testing. Series on Appl. Math. vol. 18. World Scientific (2006)

\bibitem{dyry}  D'yachkov, A.G., Rykov, V.V.: A survey of superimposed code theory.   Problems Control Inform. Theory 12, pp. 229--242  (1983)
\bibitem{erff} Erd\"os, P., Frankl, P.,  F\"uredi, Z.: Families of finite sets in which no set is covered by
the union of r others.   Israel J. Math. 51, pp. 79--89  (1985)
\bibitem{edge1} Gonen, M., Langberg, M., Sprintson A.: Group testing on general set-systems,  IEEE Trans Inf. Theory 2022, pp. 874--879 (2022)
\bibitem{optical} Harvey,  N.J.A., Patrascu, M.,   Wen, Y.,  Yekhanin, S., Chan, V.W.S.: Non-Adaptive Fault Diagnosis for All-Optical Networks via Combinatorial Group Testing on Graphs.  In:  26th IEEE Int. Conf. on Comp. Communications,
 pp. 697--705 (2007)
 

\bibitem{hl}   Hong, E.S., Ladner, R.E.: Group testing for image compression.    IEEE Transactions on Image Processing  11, no. 8,  pp. 901--911 (2002) 
 

 

\bibitem{goenka}  Goenka R., Cao S.J.,  Wong C.W.,   Rajwade A., Baron D.: Contact
tracing enhances the efficiency of covid-19 group testing. In: ICASSP 2021 
pp. 8168--8172 (2021)
\bibitem{kat} Katona G. O. H.: Intersection theorems for systems of finite sets. Acta Math. Acad. Sci. Hung.
15, pp. 329--337, (1964)
\bibitem{kts} Kautz, W.H., Singleton, R.C.: Nonrandom binary superimposed codes.    IEEE Trans Inf. Theory 10, pp. 363--377 (1964)


\bibitem{faultSensor} Lo, C., Liu, M., Lynch, J.P., Gilbert, A.C.: Efficient Sensor Fault Detection Using Combinatorial Group Testing. In:  2013 IEEE International Conference on Distributed Computing in Sensor Systems, pp. 199--206 (2013)

\bibitem{lovasz} Lov\`asz, L.: On the ratio of optimal integral
and fractional covers. Discrete Math., { 13},
383--390 (1975)
\bibitem{niko} Nikolopoulos, P.,S. Srinivasavaradhan,R., Guo T., Fragouli, C., Diggavi S.: Group testing for connected communities. In: The 24th Int. Conf. on Artificial Intelligence and Statistics.
volume 130, pp. 2341–2349. PMLR (2021)

\bibitem{rus} Ruszink\'{o}, M.: On the upper bound of the size of the $r$-cover-free families.   J.
Combin. Theory Ser. A 66, pp. 302--310 (1994)

\bibitem{sobel} Sobel M., Groll, P.A.: Group testing to eliminate efficiently all defectives in a binomial sample.   Bell System Tech. J. 38, pp. 1179--1252  (1959)


\bibitem{edge2} Vorobyev, I.:
Note on generalized group testing. Available at 
https://doi.org/10.48550/arXiv.2211.04264 (2022)



\bibitem{wolf} Wolf, J.: Born again group testing: multiaccess communications.     IEEE Trans. Inf. Theory 31, pp. 185--191  (1985)
\bibitem{zrb} Zhu, J., Rivera, K., and Baron, D.:  Noisy pooled pcr for virus testing. Available at 
https://doi.org/10.48550/arXiv.2004.02689 (2020)
\end{thebibliography}
\end{document}